%% file: main.tex
\documentclass{article}

\usepackage{microtype}
\usepackage{graphicx}
\usepackage{subfigure}
\usepackage{booktabs} 
\usepackage{amsmath}
\usepackage{amsfonts}
\usepackage{bm}
\usepackage{fixltx2e}
\usepackage{amsthm}

\usepackage{hyperref}

\usepackage[accepted]{icml2020}

\newcommand{\e}[1]{{\small $#1$}}

\newcommand{\algname}{\textsc{SpeechSplit }}
\newcommand{\algnamens}{\textsc{SpeechSplit}}

\newtheorem{theorem}{Theorem}

\newtheorem{lemma}{Lemma}[theorem]

\icmltitlerunning{Unsupervised Speech Decomposition via Triple Information Bottleneck}

\begin{document}

\twocolumn[
\icmltitle{
Unsupervised Speech Decomposition via Triple Information Bottleneck}



\icmlsetsymbol{equal}{*}

\begin{icmlauthorlist}
\icmlauthor{Kaizhi Qian}{equal,ibm,uiuc}
\icmlauthor{Yang Zhang}{equal,ibm}
\icmlauthor{Shiyu Chang}{ibm}
\icmlauthor{David Cox}{ibm}
\icmlauthor{Mark Hasegawa-Johnson}{uiuc}
\end{icmlauthorlist}

\icmlaffiliation{uiuc}{University of Illinois at Urbana-Champaign, USA}
\icmlaffiliation{ibm}{MIT-IBM Watson AI Lab, USA}

\icmlcorrespondingauthor{Yang Zhang}{yang.zhang2@ibm.com}
\icmlcorrespondingauthor{Kaizhi Qian}{kaizhiqian@gmail.com }

\icmlkeywords{Machine Learning, ICML}

\vskip 0.3in
]



\printAffiliationsAndNotice{\icmlEqualContribution} 

\begin{abstract}

\input{sections/abstract.tex}
\end{abstract}

\section{Introduction}
\label{sec:intro}

\input{sections/intro.tex}

\section{Related Work}

\textbf{The Source-Filter Model} ~Early research on speech generation proposed the source-filter model \cite{quatieri2006discrete}, and many subsequent research efforts try to decompose speech into the source that includes pitch and the filter that includes content, using signal processing approaches, such as linear predictive coding \cite{atal1979predictive}, cepstral analysis \cite{mermelstein1976distance}, temporally stable power spectral analysis \cite{kawahara2008tandem} and probabilistic approaches \cite{zhang2014improvement}. However, these approaches do not consider the prosody aspects of speech.

\textbf{Voice Conversion} ~Inspired by the style transfer and disentanglement techniques in computer vision \citep{lample2017fader, kaneko2017parallel, choi2018stargan}, many approaches based on variational autoencoders (VAEs) and generative adversarial networks (GANs) have been proposed in the field of voice conversion to disentangle the timbre information from the speech. VAE-VC \cite{hsu2016voice} directly applies VAE for voice conversion, where the encoder produces a speaker-independent content embedding. After that, VAE-GAN \cite{hsu2017voice} replaces the decoder of VAE and a GAN when generating the converted speech to improve the quality of the conversion results. CDVAE-VC \cite{huang2018voice} uses two VAEs working on different speech features, one on STRAIGHT spectra \cite{kawahara2008tandem}, and one on mel-cepstral coefficients (MCCs), and encourages that the latent representation can reconstruct both features well. ACVAE-VC \cite{kameoka2018acvae} introduces an auxiliary classifier for the conversion outputs, and encourages the converted speech to be correctly classified as the source speaker. \citet{chou2018multi} introduced a classifier for the latent code, and discourages the latent code to be correctly classified as the target speaker. Inspired by image style transfer frameworks, \citet{gao2018voice} and \citet{kameoka2018stargan} adapted CycleGan \citep{kaneko2017parallel} and StarGan \citep{choi2018stargan} respectively for voice conversion. Later, CDVAE-VC was extended by directly applying GAN \cite{huang2020unsupervised} to improve the degree of disentanglement. \citet{chou2019one} used instance normalization to further disentangle speaker from content, and thus can convert to speakers that are not seen during training. StarGan-VC2 \cite{kaneko2019stargan} refined the adversarial framework by conditioning the generator and discriminator on the source speaker label, in addition to the target speaker label. Recently, \citet{qian2019zero} proposed \textsc{AutoVC}, a simple autoencoder based method that disentangles the timbre and content using information-constraining bottlenecks. Later, \citet{qian2020f0} fixed the pitch jump problem of \textsc{AutoVC} by F0 conditioning. Besides, the time-domain deep generative model is gaining more research attention for voice conversion \cite{niwa2018statistical,nachmani2019unsupervised,serra2019blow}.  However, these methods only focus on converting timbre, which is only one of the speech components.

\textbf{Prosody Disentanglement} 
There are recently many text-to-speech (TTS) systems that seek to disentangle the prosody information to generate expressive speech.
~\citet{skerry2018towards} introduced a Tacotron based speech synthesizer that can disentangle prosody from speech content by having an encoder that can extract the prosody information from the original speech. Mellotron \citep{valle2020mellotron} is a speech synthesizer conditional on both explicit prosody labels and latent prosody code to capture and disentangle different aspects of the prosody information. CHiVE \citep{kenter2019chive} introduces a hierarchical encoder-decoder structure that is conditioned on a set of prosodic features and linguistic features.
However, these TTS systems all require text transcriptions, which, as discussed, makes the task easier but limits their applications to high-resource language.
Besides TTS systems, Parrotron \cite{biadsy2019parrotron} disentangles prosody by encouraging the latent codes to be the same as the corresponding phone representation of the input speech. However, Parrotron still requires text transcriptions to label the phone representation, as well as to generate the synthetic parallel dataset. \citet{polyak2019attention} proposed, to the best of our knowledge, the only prosody disentanglement algorithm that does not rely on text transcriptions, which attempts to remove the rhythm information by randomly resampling the input speech. However, the effect of their prosody conversion is not very pronounced. In this paper, we would like to achieve effective prosody conversion without using text transcriptions, which is more flexible for low-resource languages.

\input{sections/background}
\section{\algnamens}

This section introduces \algname. For notation, upper-cased letters, \e{X} and \e{\bm X}, denote random scalars and vectors respectively; lower-cased letters, \e{x} and \e{\bm x}, denote deterministic scalars and vectors respectively; \e{H(\bm X)} denotes the Shannon entropy of \e{\bm X}; \e{H(Y | \bm X)} denotes the entropy of \e{Y} conditional on \e{\bm{X}}; \e{I(Y; \bm X)} denotes the mutual information.

\subsection{Problem Formulation}

Denote \e{\bm S=\{\bm S_t\}} as a speech spectrogram, where \e{t} is the time index. Denote the speaker's identity as \e{U}. We assume that \e{\bm S} and \e{U} are generated through the following random generative processes
\begin{equation}
\small
    \bm S = g_s(\bm C, \bm R, \bm F, \bm V), \quad U = g_u(\bm V),
    \label{eq:speech_gen}
\end{equation}
where \e{\bm{C}} denotes content; \e{\bm R} denotes rhythm; \e{\bm F} denotes pitch target; \e{\bm V} denotes timbre. \e{g_s(\cdot)} and \e{g_u(\cdot)} are assumed to be a one-to-one mapping. Note that here we assume \e{\bm C} also accounts for the residual information that is not included in rhythm, pitch or timbre.

Our goal is to construct an autoencoder-based generative model for speech, such that the hidden code contains disentangled representations of the speech components. We formally denote the representations as \e{\bm Z_c}, \e{\bm Z_r} and \e{\bm Z_f}, and these representations should satisfy
\begin{equation}
    \small
    \bm Z_c = h_c(\bm C), \quad \bm Z_r = h_r(\bm R), \quad \bm Z_f = h_f(\bm F),
    \label{eq:disentagle}
\end{equation}
where \e{h_c(\cdot)}, \e{h_r(\cdot)} and \e{h_f(\cdot)} are all one-to-one mappings.



\subsection{\textsc{AutoVC} and Its Limitations}
\label{subsec:autovc}

Since \algname inherits the information bottleneck mechanism proposed in \textsc{AutoVC}, it is necessary to first review its framework and limitations. Figure~\ref{fig:model}(a) shows the framework of \textsc{AutoVC}, which consists of an encoder and a decoder. The encoder has an information bottleneck at the end (shown as the grey tip), which is implemented as hard constraint on code dimensions. The input to the encoder is speech spectrogram \e{\bm S}, and the output of the encoder is called the speech code, denoted as \e{\bm Z}. The decoder takes \e{\bm Z} and the speaker identity label \e{U} as its inputs, and produces a speech spectrogram \e{\hat{\bm S}} as output. We formally denote the encoder as \e{E(\cdot)}, and the decoder as \e{D(\cdot, \cdot)}. The \textsc{AutoVC} pipeline can be expressed as
\begin{equation}
    \small
    \bm Z = E(\bm S), \quad \hat{\bm S} = D(\bm Z, U).
\end{equation}
During training, the output of the decoder tries to reconstruct the input spectrogram:
\begin{equation}
    \small
    \min_{\bm \theta} \mathbb{E}[\Vert \hat{\bm S} - \bm S \Vert _2^2],
    \label{eq:loss}
\end{equation}
where \e{\bm \theta} denotes all the trainable parameters.

It can be shown that if the information bottleneck is tuned to the right size, this simple scheme can achieve disentanglement of the timbre information as
\begin{equation}
    \small
    \bm Z = h(\bm C, \bm R, \bm F).
\end{equation}
Figure~\ref{fig:model}(a) provides an intuitive explanation of why this is possible. As can be seen, speech is represented as a concatenation of different blocks, indicating the content, rhythm, pitch and timbre information. Note that speaker identity is represented with the same block style as timbre because it is assumed to preserve equivalent information to timbre according to equation \eqref{eq:speech_gen}. Since the speaker identity is separately fed to the decoder, the decoder can still have access to all the information to perform self-reconstruction even if the encoder does not preserve the timbre information in its output. Therefore, when the information bottleneck is binding, the encoder will remove the timbre information. However, \e{\bm Z} still lumps content, rhythm, and pitch together. As a result, \textsc{AutoVC} can only convert timbre.


\subsection{The \algname Framework}

Figure \ref{fig:model}(b) illustrates the \algname framework. \algname is also an autoencoder with an information bottleneck. However, in order to further decompose the remaining speech components, \algname introduces three encoders with heterogeneous information bottleneck, which are a \emph{content encoder}, a \emph{rhythm encoder}, and a \emph{pitch encoder}. Below are the details of the encoders and the decoder of \algnamens.

\textbf{The Encoders} ~As shown in figure~\ref{fig:model}(b), all three encoders are almost the same, but with two subtle differences. First, the input to the content encoder and rhythm encoder is speech \e{\bm S}, whereas the input to the pitch encoder is the normalized pitch contour, which we denote as \e{\bm P}. As discussed in section~\ref{sec:background}, the normalized pitch contour \e{\bm P} refers to the pitch contour that is normalized to have the same mean and variance across all the speakers, so the normalized pitch contour only contains the pitch information, \e{\bm F}, and rhythm information, \e{\bm R}, but no speaker ID information, \e{U}.

Second, the content encoder and pitch encoder perform a random resampling operation along the time dimension of the input. Random resampling involves two steps of operations. The first step is to divide the input into segments of random lengths. The second step is to randomly stretch or squeeze each segment along the time dimension. Therefore, random resampling can be regarded as an information bottleneck on rhythm. More details of random resampling can be found in appendix~\ref{subsec:info_bottleneck}. All the encoders have the physical information bottleneck at the output. The final outputs of the encoders are called content code, rhythm code and pitch code, which are denoted as \e{\bm Z_c}, \e{\bm Z_r} and \e{\bm Z_f} respectively. Formally, denote the content encoder as \e{E_c(\cdot)}, rhythm encoder as \e{E_r(\cdot)} and pitch encoder as \e{E_f(\cdot)}, and denote the random resampling operation as \e{A(\cdot)}. Then we have
\begin{equation}
\small
    \bm Z_c = E_c (A (\bm S)), \quad \bm Z_r = E_r(\bm S), \quad \bm Z_f = E_f (A(\bm P)).
\end{equation}

\textbf{The Decoder} ~The decoder takes all the speech code and the speaker identity label (or embedding) as its inputs, and produce a speech spectrogram as output, \emph{i.e.},
\begin{equation}
    \small
    \hat{\bm S} = D(\bm Z_c, \bm Z_r, \bm Z_f, U).
\end{equation}
During training, the output of the decoder tries to reconstruct the input spectrogram, which is the same as in equation \eqref{eq:loss}.

Counter-intuitive as it may sound, we claim that when all the information bottlenecks are appropriately set and the network representation power is sufficient, a minimizer of equation \eqref{eq:loss} will satisfy the disentanglement condition as in equation \eqref{eq:disentagle}. In what follows, we will explain why such decomposition is possible.

\begin{figure}[t]
\raggedleft
\includegraphics[width=1\linewidth]{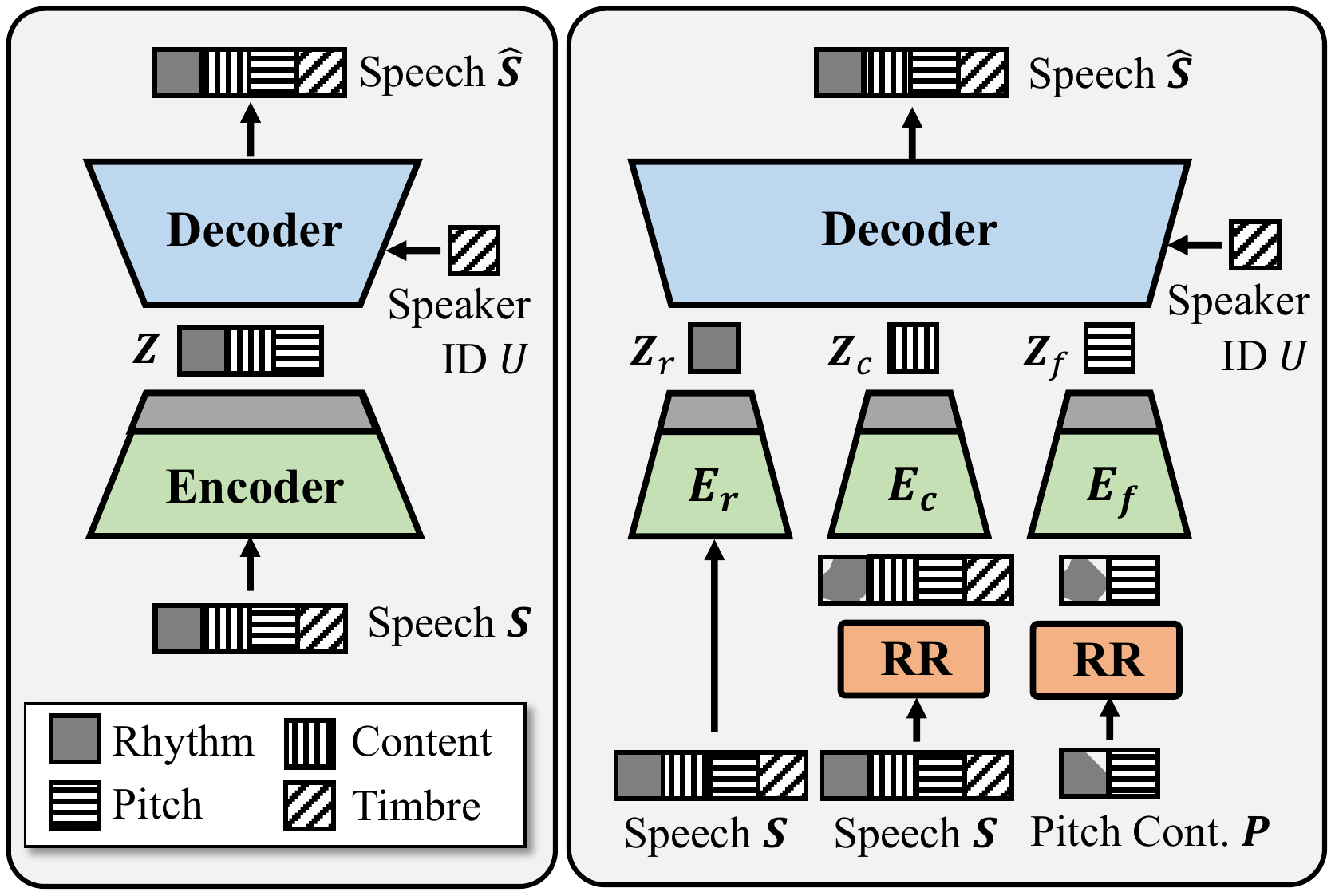}
\small
(a) \textsc{AutoVC}~~~~~~~~~~~~~~~~~~~~~~~~~~~(b) \algnamens~~~~~~~~~~~~~~~
\caption{\small{Frameworks and \textsc{AutoVC} and \algname and illustration of why they can perform disentanglement. Signals are represented as blocks to denote their information components. \e{E_r} denotes the rhythm encoder; \e{E_c} denotes the content encoder; \e{E_f} denotes the pitch encoder. `RR' denotes random resampling. `Pitch Cont.' is short for the normalized pitch contour. The grey block at the tip of the encoders denotes the information bottleneck. Some rhythm blocks have some holes in them, which represents that a portion of the rhythm information is lost. The bottlenecks force the encoders to pass only the information that other encoders cannot supply, hence achieving the disentanglement.}}
\label{fig:model}
\end{figure}


\subsection{Why Does It Force Speech Decomposition?}
\label{subsec:why}

Figure~\ref{fig:model} provides an intuitive illustration of how \algname achieves speech decomposition, where a few important assumptions are made.

\textit{Assumption 1:} The random resampling operation will contaminate the rhythm information \e{\bm R}, \emph{i.e.} \e{\forall \bm r_1 \neq \bm r_2}
\begin{equation}
\small
    Pr[A(g_s(\bm C, \bm r_1, \bm F, \bm V)) = A(g_s(\bm C, \bm r_2, \bm F, \bm V))]>0.
    \label{eq:resample1}
\end{equation}

\textit{Assumption 2:} The random resampling operation will not contaminate the other speech components, \emph{i.e.}
\begin{equation}
\small
    I(\bm C; A(\bm S)) = H(\bm C), \quad I(\bm F; A(\bm S)) = H(\bm F).
    \label{eq:resample2}
\end{equation}

\textit{Assumption 3:} The pitch contour \e{\bm P} contains all the pitch information and a portion of rhythm information.
\begin{equation}
    \small
    \bm P = g_p(\bm F, \bm R), \quad I(\bm F; \bm P) = H(\bm F).
    \label{eq:pitch_contour}
\end{equation}

As shown in figure~\ref{fig:model}(b), speech contains four blocks of information. When it passes through the random resampling operation, a random portion of the rhythm block is wiped (shown as the holes in the rhythm block at the output of the RR module), but the other blocks remain intact. On the other hand, the normalized pitch contour mainly contains two blocks, the pitch block, and the rhythm block. The rhythm block is missing a corner because the normalized pitch contour does not contain all the rhythm information, and it misses even more when it passes through the random resampling module.

Similar to the \textsc{AutoVC} claim, the timbre information is directly fed to the decoder, so all the encoders do not need to encode the timbre information. Therefore, this section focuses on explaining why \algname can force the encoders to separately encode the content, pitch, and timbre. 

First, the rhythm encoder \e{E_r(\cdot)} is the only encoder that has access to the complete rhythm information \e{\bm R}. The other two encoders only preserve a random portion of \e{\bm R}, and there is no way for \e{E_r(\cdot)} to guess which part is lost and thus only supply the lost part. Therefore, \e{E_r(\cdot)} must pass all the rhythm information. Meanwhile, the other aspects are available in the other two encoders. So if \e{E_r(\cdot)} is forced to lose some information by its information bottleneck, it will prioritize removing the content, pitch, and timbre.

Second, given that \e{E_r(\cdot)} only encodes \e{\bm R}, then the content encoder \e{E_c(\cdot)} becomes the only encoder that can encode all the content information \e{\bm C}, because the pitch encoder does not have access to \e{\bm C}. Therefore, \e{E_c(\cdot)} must pass all the content information. Meanwhile, the other aspects can be supplied elsewhere, so the rhythm encoder will remove the other aspects if the information bottleneck is binding.

Finally, with \e{E_r(\cdot)} encoding only \e{\bm R} and \e{E_c(\cdot)} encoding only \e{\bm C}, the pitch encoder \e{E_f(\cdot)} must encode the pitch information. All the other aspects are supplied in other channels, so \e{E_f(\cdot)} will prioritize removing these aspects if the information bottleneck is binding.

Simply put, if each encoder is only allowed to pass one block, then the arrangement in figure~\ref{fig:model} is the only way to ensure full recovery of the speech information. 

Now we are ready to give our formal result.
\begin{theorem}
Assume \e{\bm C}, \e{\bm R}, \e{\bm F}, \e{U} are independent, and that the information bottleneck is precisely set such that
\begin{equation}
    \small
    H(\bm Z_c) = H(\bm C), \quad H(\bm Z_r) = H(\bm R), \quad H(\bm Z_f) = H(\bm F).
    \label{eq:info_assump}
\end{equation}
Assume equations \eqref{eq:speech_gen}, \eqref{eq:resample1}, \eqref{eq:resample2} and \eqref{eq:pitch_contour} hold. Then the global optimum of equation \eqref{eq:loss} would produce the disentangled representation as in equation \eqref{eq:disentagle}.
\label{thm:main}
\end{theorem}
The proof is presented in the appendix. Although theorem \ref{thm:main} is contingent on a set of relatively stringent conditions, which may not hold in practice, we will empirically verify the disentanglement capabilities in section~\ref{sec:exper}.

\subsection{Network Architecture}

\begin{figure}[t]
\centering
\includegraphics[width=1\linewidth]{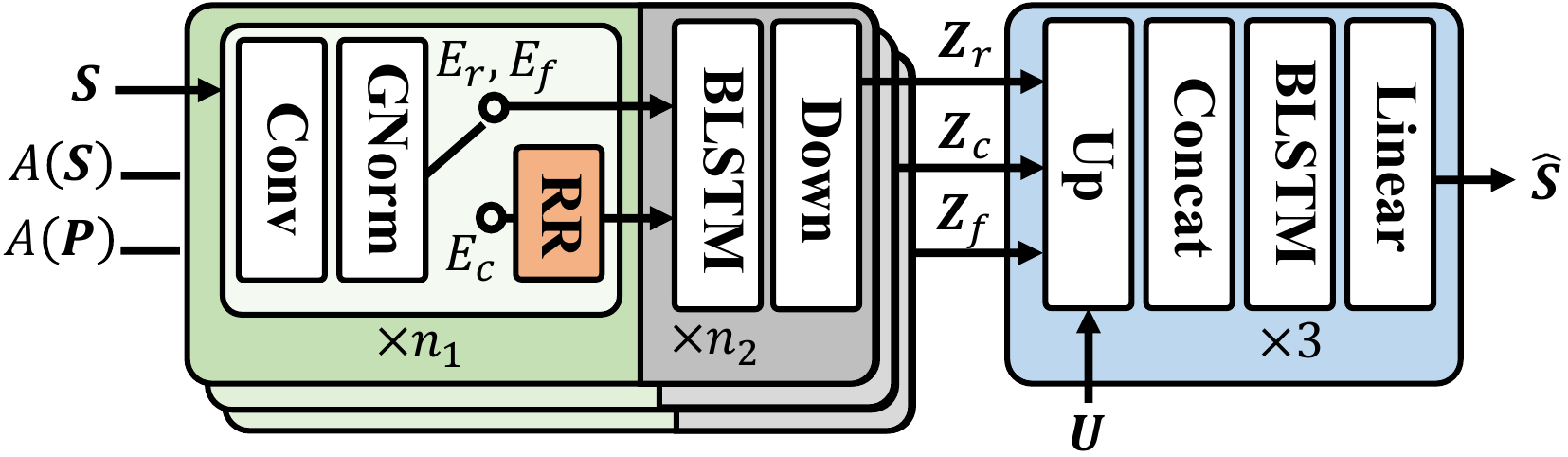}
\small
\caption{\small{The architecture of \algnamens. `GNorm' denotes group normalization; `RR' denotes random resampling; `Down' and `Up' denote downsampling and upsampling operations respectively. `Linear' denotes linear projection layer. \e{\times n} denotes the module above is repeated \e{n} times.}}
\label{fig:archi}
\end{figure}

Figure~\ref{fig:archi} shows the architecture of \algnamens. The left module corresponds to the encoders and the right to the decoder. All three encoders share a similar architecture, which consists of a stack of \e{5\times 1} convolutional layers followed by group normalization \cite{wu2018group}. For the content encoder, the output of each convolutional layer is passed to a random resampling module to further contaminate rhythm. The final output of the convolutional layers is fed to a stack of bidirectional-LSTM layers to reduce the feature dimension, which is then passed through a downsampling operation to reduce the temporal dimension, producing the hidden representations. Table~\ref{tab:hyper} shows the hyperparameter settings of each encoder. More details of the downsampling operation are provided in appendix~\ref{subsec:info_bottleneck}.

\begin{table}[t]
    \centering
    \small
    \caption{Hyperparameter settings of the encoders.}
    
    \begin{tabular}{l|lll}
        \hline\hline
         & \textbf{Rhythm} & \textbf{Content} & \textbf{Pitch}\\
        \hline
        Conv Layers & 1 & 3 & 3 \\
        Conv Dim & 128 & 512 & 256 \\
        Norm Groups & 8 & 32 & 16 \\
        BLSTM Layers & 1 & 2 & 1 \\
        BLSTM Dim & 1 & 8 & 32 \\
        Downsample Factor & 8 & 8 & 8 \\
         \hline\hline
    \end{tabular}
    \label{tab:hyper}
\end{table}

The decoder first upsamples the hidden representation to restore the original sampling rate. The speaker identity label \e{U}, which is a one-hot vector, is also repeated along the time dimension to match the temporal dimension of the other upsampled representations. All the representations are then concatenated along the channel dimension and fed to a stack of three bidirectional-LSTM layers with an output linear layer to produce the final output. The spectrogram is converted back to the speech waveform using the same wavenet-vocoder as in \textsc{AutoVC}. Additional architecture and implementation details are provided in appendix~\ref{sec:implementation_details}.


\input{sections/exp}

\section{Conclusion}

We have demonstrated that \algname has powerful disentanglement capabilities by having multiple intricately designed information bottlenecks. There are three takeaways. First, we have shown that the physical dimension of the hidden representations can effectively limit the information flow. Second, we have verified that when information bottleneck is binding, neural autoencoder will only pass the information that other channels cannot provide. Third, even if we only have a partial disentanglement algorithm, \emph{e.g.} the random resampling, we can still design a complete disentanglement algorithm by having multiple channels with different information bottleneck. These intriguing observations inspire a generic approach to disentanglement. As future directions, we will seek to refine the bottleneck design of \algname using more information-theory-guided approaches, as well as explore the application of \algname to low-resource speech processing systems.

\section*{Acknowledgment}

We would like to give special thanks to Gaoyuan Zhang from MIT-IBM Watson AI Lab, who has helped us a lot with building our demo webpage.


\bibliography{example_paper}
\bibliographystyle{icml2020}

\newpage
\newpage
\appendix
\input{sections/appendix.tex}

\end{document}

%% file: sections/abstract.tex
Speech information can be roughly decomposed into four components: language content, timbre, pitch, and rhythm. Obtaining disentangled representations of these components is useful in many speech analysis and generation applications. Recently, state-of-the-art voice conversion systems have led to speech representations that can disentangle speaker-dependent and independent information. However, these systems can only disentangle timbre, while information about pitch, rhythm and content is still mixed together. Further disentangling the remaining speech components is an under-determined problem in the absence of explicit annotations for each component, which are difficult and expensive to obtain. In this paper, we propose \algnamens, which can blindly decompose speech into its four components by introducing three carefully designed information bottlenecks. \algname is among the first algorithms that can separately perform style transfer on timbre, pitch and rhythm without text labels. Our code is publicly available at {\small \url{https://github.com/auspicious3000/SpeechSplit}}.

%% file: sections/intro.tex
Human speech conveys a rich stream of information, which can be roughly decomposed into four important components: \emph{content}, \emph{timbre}, \emph{pitch} and \emph{rhythm}. The language content of speech comprises the primary information in speech, which can also be transcribed to text. Timbre carries information about the voice characteristics of a speaker, which is closely connected with the speaker's identity. Pitch and rhythm are the two major components of \emph{prosody}, which expresses the emotion of the speaker. Pitch variation conveys the aspects of the tone of the speaker, and rhythm characterizes how fast the speaker utters each word or syllable.

For decades, speech researchers have sought to obtain disentangled representations of these speech components, which are useful in many speech applications. In speech analysis tasks, the disentanglement of speech components helps to remove interference introduced by irrelevant components. In speech generation tasks, disentanglement is the foundation of many applications, such as voice conversion \cite{chou2019one}, prosody modification \cite{shechtman2019sequence}, emotional speech synthesis \cite{pell2011emotional}, and low bit-rate speech encoding \cite{schroeder1985code}, to name a few. 

Recently, state-of-the-art voice conversion systems have been able to obtain a speaker-invariant representation of speech, which disentangles the speaker-dependent information \cite{qian2019zero,chou2018multi,chou2019one}. However, these algorithms are only able to disentangle timbre. The remaining aspects, \emph{i.e.} content, pitch, and rhythm are still lumped together. As a result, the converted speech produced by these algorithms differs from the source speech only in terms of timbre. The pitch contour and rhythm remain largely the same.

From an information-theoretic perspective, the success in timbre disentanglement can be ascribed to the availability of a speaker identity label, which preserves almost all the information of timbre, such that voice conversion systems can `subtract' such information from speech. For example, \textsc{AutoVC} \cite{qian2019zero}, a voice conversion system, constructs an autoencoder for speech and feeds the speaker identity label to the decoder. As shown in figure \ref{fig:model}(a), by constructing an information bottleneck between the encoder and decoder, \textsc{AutoVC} can force the encoder to remove the timbre information, because the equivalent information is supplied to the decoder directly. It is worth noting that although the speaker identity is also correlated with the pitch and timbre information, the information overlap is relatively small, so the speaker identity cannot serve as labels for pitch and rhythm. If we had analogous information-preserving labels for timbre, rhythm or pitch, the disentanglement of these aspects would be straightforward, simply by utilizing these labels the same way voice conversion algorithms use the speaker identity label.

However, obtaining annotations for these other speech components is challenging. First, although language content annotation is effectively provided by text transcriptions, obtaining a large number of text transcriptions is expensive, especially for low-resourced languages. Therefore, here, we will focus on unsupervised methods that do not rely on text transcriptions. Second, the rhythm annotation, which is essentially the length of each syllable, can only be obtained with the help of text transcriptions, which are again unavailable under our unsupervised setting. Finally, for pitch annotation, although the pitch information can be extracted as pitch contour using pitch extraction algorithms, the pitch contour itself is entangled with rhythm information, because it contains the information of how long each speech segment is. Without the information preserving labels, disentangling content, rhythm and pitch becomes an under-determined problem. Hence, here we ask: is it possible to decompose these remaining speech components without relying on text transcriptions and other information-preserving labels?


In this paper, we propose \algnamens, a speech generative model that can blindly decompose speech into content, timbre, pitch, and rhythm, and generate speech from these disentangled representations. Thus, \algname is among the first algorithms that can enable flexible conversion of different aspects to different styles without relying on any text transcription. To achieve unsupervised decomposition, \algname introduces an encoder-decoder structure with three encoder channels, each with a different, carefully-crafted information bottleneck design. The information bottleneck is imposed by two mechanisms: first, a constraint on the physical dimension of the representation, which has been shown effective in \textsc{AutoVC}, and second, the introduction of noise by randomly resampling along the time dimension, which has been shown effective in \cite{polyak2019attention}. We find that subtle differences in the information bottleneck design can force different channels to pass different information, such that one passes language content, one passes rhythm, and one passes pitch information,  thereby achieving the blind disentanglement of all speech components. 

Besides direct value in speech applications, \algname also provides insight into a powerful design principle that can be broadly applied to any disentangled representation learning problem: in the presence of an information bottleneck, a neural network will prioritize passing through the information that cannot be provided elsewhere. This observation inspires a generic approach to disentanglement.

%% file: sections/background.tex
\section{Background: Information in Speech}
\label{sec:background}

\begin{figure}[t]
\centering
\includegraphics[width=1\linewidth]{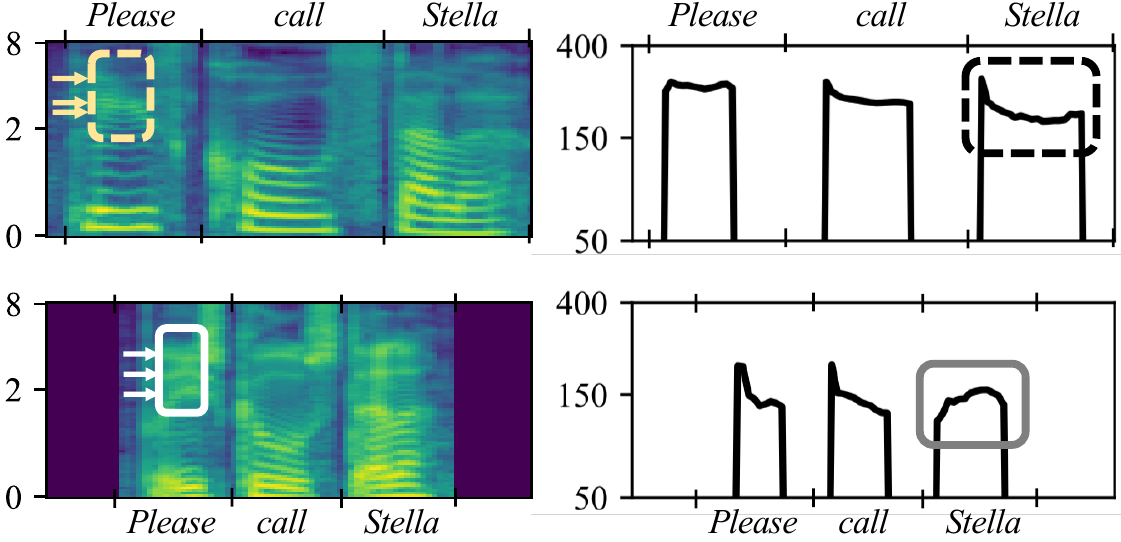}
\small
\caption{\small{Spectrograms (left) and pitch contours (right) of two utterances of the same sentence \textit{`Please call Stella'}. The left rectangle marks highlight the formant structures of the phone \textit{`ea'}. The arrows mark the frequencies of the second, third and fourth formants. The right rectangle marks highlight the pitch tones of the word \textit{`Stella'}. }}
\label{fig:background}
\end{figure}

Since this paper focuses on the decomposition of speech information into rhythm, pitch, timbre, and content, we provide here a brief primer on each of these components.  Figure~\ref{fig:background} shows the spectrograms (left) and pitch contours (right) of utterances of the sentence \textit{`Please call Stella'}. Throughout this paper, the term `spectrogram' refers to the magnitude spectrogram.

\textbf{Rhythm}  ~Rhythm characterizes how fast the speaker utters each syllable. As shown in figure~\ref{fig:background}, each spectrogram is divided into segments, which correspond to each word, as marked on the horizontal axis. So the lengths of these segments reflect the rhythm information. In the top spectrogram, each segment is long, indicating a slow speaker; in the bottom spectrogram, each segment is short, indicating a fast speaker.

\textbf{Pitch} ~Pitch is an important component of intonation. One popular representation of the pitch information is the \emph{pitch targets} \citep{xu2001pitch}, which is defined as the intended pitch, \emph{e.g.} rise or fall, high or low \emph{etc.}, of each syllable. The pitch information, or pitch target information, is contained in the pitch contour, because the pitch contour is generally considered as the result of a constant attempt to hit the pitch targets of each syllable, subject to the physical constraints \cite{xu2001pitch}. However, the pitch contour also entangles other information. First, the pitch contour contains the rhythm information, because each nonzero segment of the pitch contour represents a voiced segment, which typically corresponds to a word or a syllable. So the length of each voiced segment indicates how fast the speaker speaks. Second, the pitch range reflects certain speaker identity information -- female speakers tend to have a high pitch range, as shown in the upper panel of figure~\ref{fig:background}, and male speakers tend to have a low pitch range, as shown in the lower panel of figure~\ref{fig:background}. Here, we assume that the impact of the speaker identity on the pitch contour is linear. In other words, if we normalize the pitch contour to a common mean and standard deviation, the speaker identity information will be removed. To sum up, the pitch contour entangles the information of speaker identity, rhythm and pitch; the normalized pitch contour only contains the information of the latter two.

\textbf{Timbre} ~Timbre is perceived as the voice characteristics of a speaker. It is reflected by the \emph{formant frequencies}, which are the resonant frequency components in the vocal tract. In a spectrogram, the formants are shown as the salient frequency components of the spectral envelope. In figure~\ref{fig:background}, the rectangles and arrows on the spectrogram highlight three formants. As can be seen, the top spectrogram has a higher formant frequency range, indicating a bright voice; the bottom spectrogram has a lower formant frequency range, indicating a deep voice.

\textbf{Content} ~In English and many other languages, the basic unit of content is the phone. Each phone comes with a particular formant pattern. For example, the three formants highlighted in figure\ref{fig:background} are the second, third and fourth lowest formants of the phone \textit{`ea'} as in \textit{`please'}. Although their formant frequencies have different ranges, which indicates their timbre difference, they have the same pattern -- they tend to cluster together and are far away from the lowest formant (which is at around 100 Hz). 

%% file: sections/exp.tex
\section{Experiments}
\label{sec:exper}

In this section, we will empirically verify the disentanglement capability of \algnamens.  We will be visualizing our speech results using spectrogram and pitch contour.  However, to fully appreciate the performance of \algnamens, we strongly encourage readers to refer to our online demo\footnote{\url{https://auspicious3000.github.io/SpeechSplit-Demo/}}. Additional experiment results can be found in appendix~\ref{sec:exper_add}. The frequency axis units of all the spectrograms are in kHz, and those of the pitch contour plots are in Hz.

\subsection{Configurations}

The experiments are performed on the VCTK dataset \cite{veaux2016superseded}. The training set contains 20 speakers where each speaker has about 15 minutes of speech. The test set contains the same 20 speakers but with different utterances, which is the conventional voice conversion setting. \algname is trained using the ADAM optimizer \cite{kingma2014adam} with a batch size of 16 for 800k steps. Since there are no other algorithms that can perform blind decomposition so far, we will be comparing our result with \textsc{AutoVC}, a conventional voice conversion baseline.

The model selection is performed on the training dataset. Specifically, the physical bottleneck dimensions are tuned based on the criterion: when the input to one of the encoders or the speaker embedding is set to zero, the output reconstruction should not contain the corresponding information. As will be shown in section~\ref{subsec:zero}, setting the inputs and speaker embedding to zero can measure the degree of disentanglement. From the models that satisfy this criterion, we pick the one with the lowest training error. Appendix~\ref{app:tuning} provides additional guidance on how to tune the bottlenecks.

\subsection{Rhythm, Pitch and Timbre Conversions}
\label{subsec:main_conversion}

If \algname can decompose the speech into different components, then it should be able to separately perform style transfer on each aspect, which is achieved by replacing the input to the respective encoder with that of the target utterance. For example, if we want to convert pitch, we feed the target pitch contour to the pitch encoder. To convert timbre, we feed the target speaker id to the decoder.

We construct parallel speech pairs from the test set, where both the source and target speakers read the same utterances. Please note that we use the parallel pairs \emph{only for testing}. During training, \algname is trained without parallel speech data. For each parallel pair, we set one utterance as the source and one as the target, and perform seven different types of conversions, including three single-aspect conversions (rhythm-only, pitch-only and timbre-only), three double-aspect conversions (rhythm+pitch, rhythm+timbre, and pitch+timbre), and one all-aspect conversion.

\textbf{Conversion Visualization} ~Figure~\ref{fig:main_spect} shows the single-aspect conversion results on a speech pair uttering \textit{`Please call Stella'}. The source speaker is a slow female speaker, and the target speaker is a fast male speaker. As shown in figure~\ref{fig:main_spect}, \algname can separately convert each aspect. First, in terms of rhythm, note that the rhythm-only conversion is perfectly aligned with the target utterance in time, whereas the timbre-only and pitch-only conversions are perfectly aligned with the source utterance in time. Second, in terms of pitch, notice that the timbre-only and rhythm-only conversions have a falling tone on the word \textit{`Stella'}, which is the same as the source utterance, as highlighted by the dashed rectangle. The pitch-only conversion has a rising tone on \textit{`Stella'}, which is the same as the target utterance, as highlighted by the solid rectangles. Third, in terms of timbre, as highlighted by the rectangles on the spectrograms, the formants of pitch-only and rhythm-only conversions are as high as those of the source speech, and the formants of timbre-only conversions are as high as those in the target.

\begin{figure}[t]
\raggedleft
\includegraphics[width=1\linewidth]{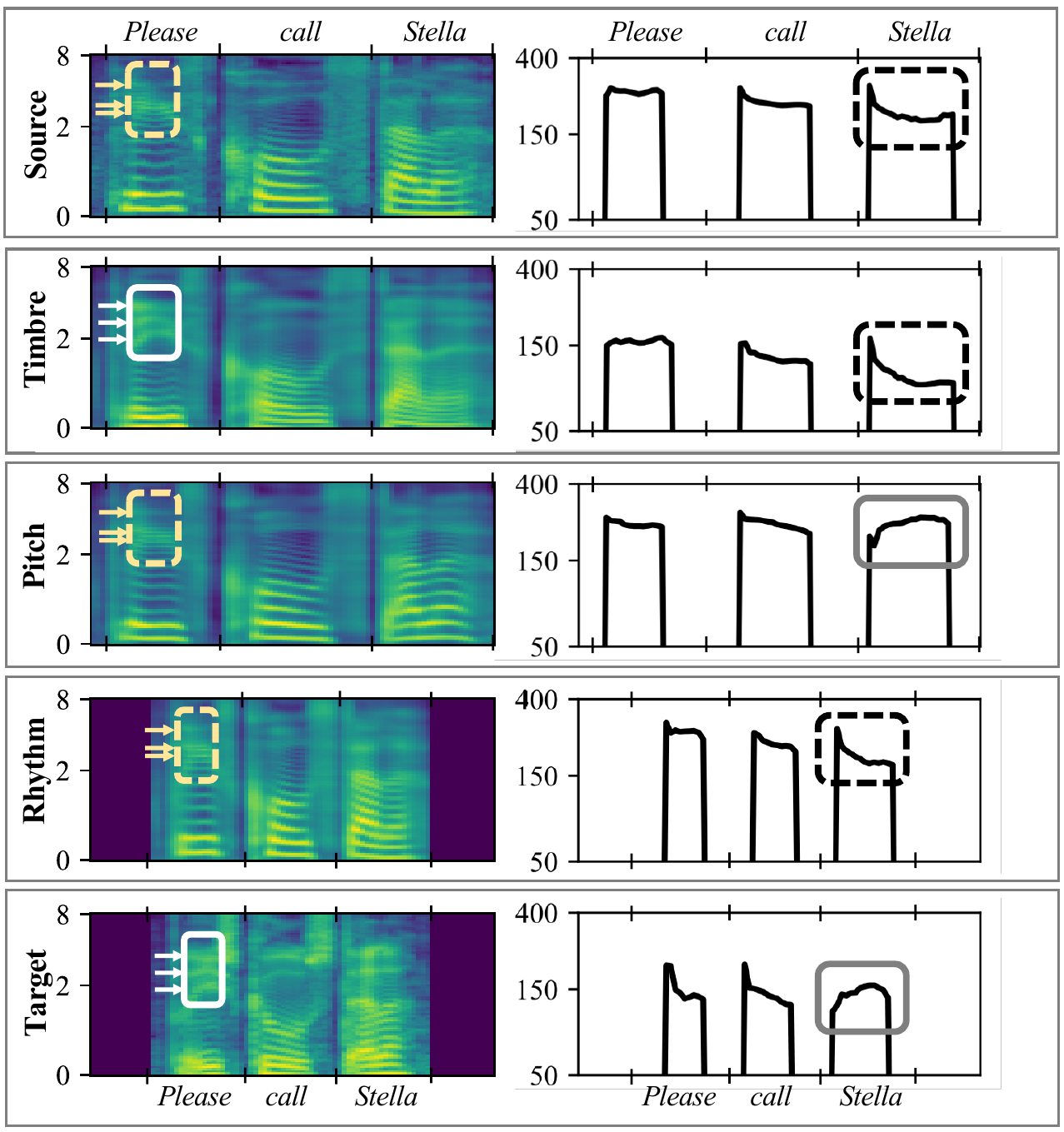}
\small
\vspace{-0.1in}
\caption{\small{Spectrogram (left) and pitch contours (right) of single-aspect conversion results of the utterance \textit{`Please call Stella'}. The left rectangle marks highlight the formant structures of the phone \textit{`ea'}. The arrows mark the frequencies of the second, third and fourth formants. The right rectangle marks highlight the pitch tones of the word \textit{`Stella'}. }}
\label{fig:main_spect}
\end{figure}

\textbf{Subjective Evaluation} ~We also perform a subjective evaluation on \emph{Amazon Mechanical Turk} on whether the conversion of each aspect is successful. For example, to evaluate whether the different conversions convert pitch, we select 20 speech pairs that are perceptually distinct in pitch, and generate all the seven types of conversions, plus the \textsc{AutoVC} conversion and the source utterance as baselines. Each test is assigned to five subjects. In the test, the subject is presented with two reference utterances, which are the source and target utterances in a random order, and then with one of the nine conversion results. The subject is asked to select which reference utterance has a more similar pitch tone to the converted utterance. We compute the \emph{pitch conversion rate} as the percentage of answers that choose the target utterance. We would expect the utterances with pitch converted to have a high pitch conversion rate; otherwise, the pitch conversion rate should be low. The \emph{rhythm conversion rate} and \emph{timbre conversion rate} are computed in a similar way. More details of the test token generation process can found in appendix \ref{app:test_select}.

\begin{figure}[t]
\centering
\includegraphics[width=\linewidth]{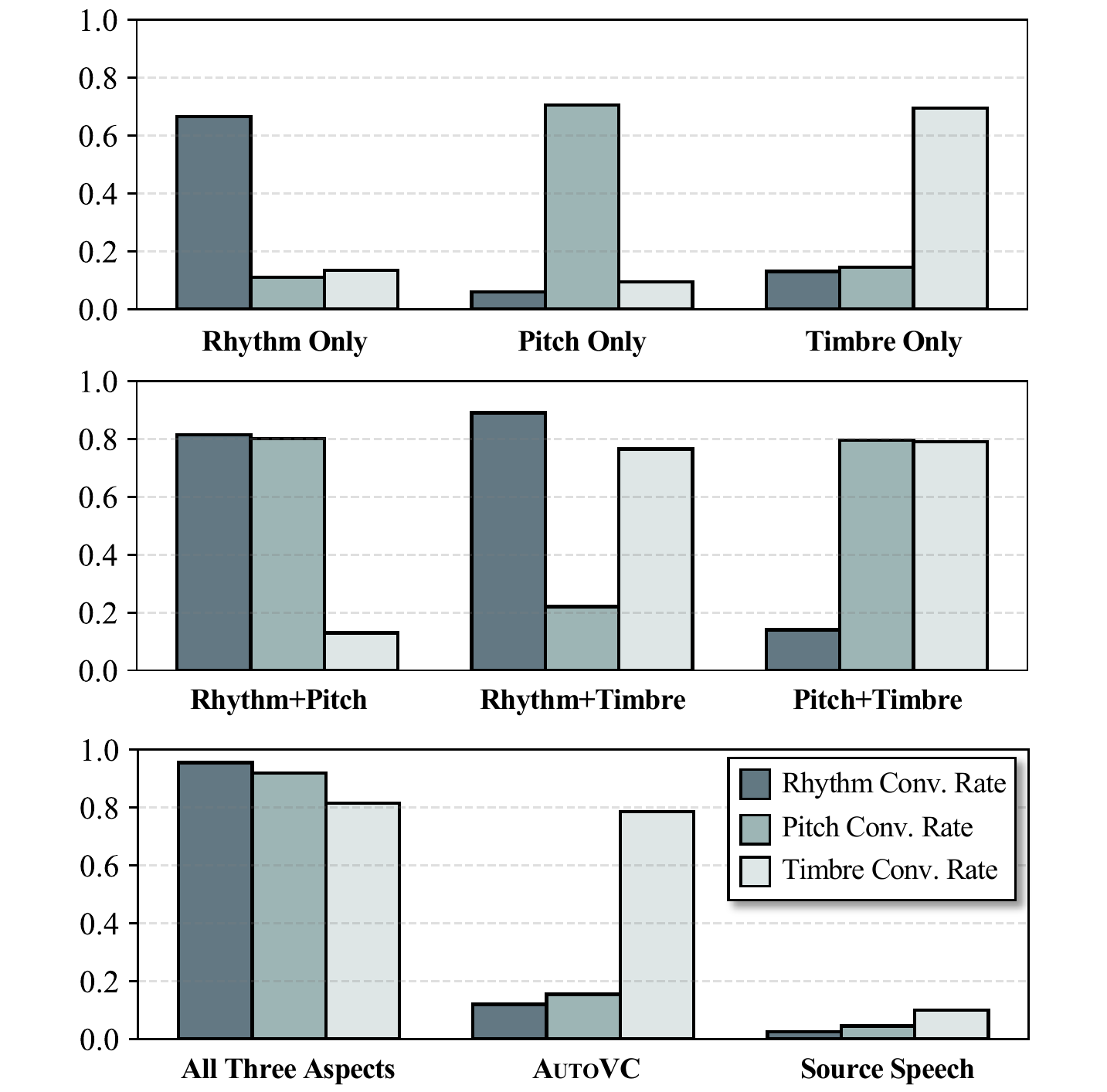}
\small
\vspace{-0.1in}
\caption{\small{Subjective conversion rates of different conversion types. Each bar group corresponds to a conversion type/algorithm. The three bars within each group represent the rhythm, pitch and timbre conversion rates respectively.}}
\label{fig:subject}
\end{figure}

\begin{figure}[ht]
\centering
\includegraphics[width=\linewidth]{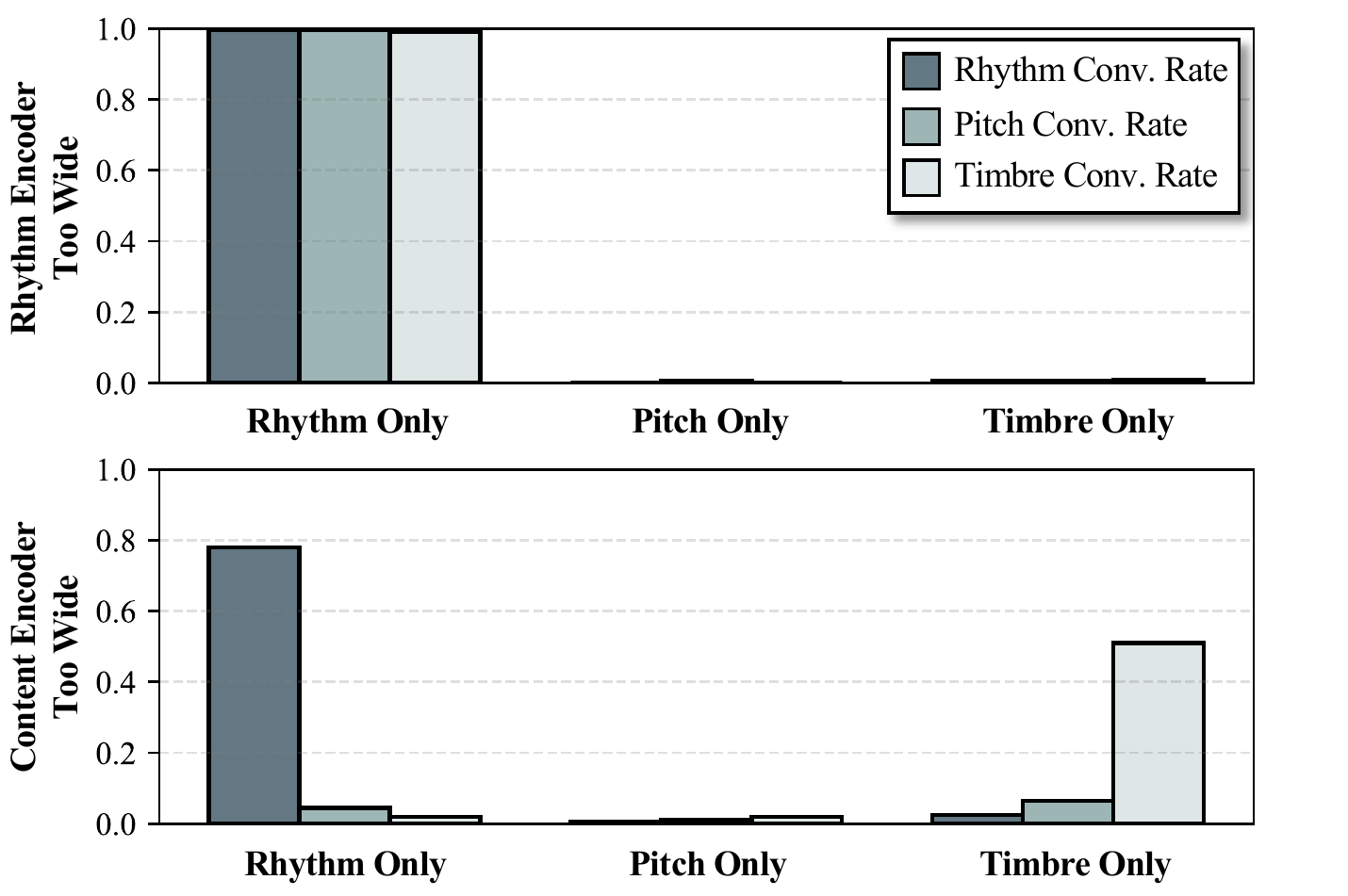}
\small
\vspace{-0.1in}
\caption{\small{Subjective conversion rates of single-aspect conversions of \algname when the information bottleneck of the rhythm encoder (top panel) or the content encoder (bottom panel) is too wide. Each bar group corresponds to a conversion type/algorithm. The three bars within each group represent the rhythm, pitch and timbre conversion rates respectively.}}
\label{fig:subject_wide}
\end{figure}

\begin{table}[t]
\label{mos}
\centering
\vspace{-0.1in}
\small
\caption{MOS of different conversion types/algorithms.}
\begin{tabular}{ccc}
    \hline\hline
     \textbf{Rhythm Only}& \textbf{Pitch Only} & \textbf{Timbre Only}  \\
     3.21 & 3.79 & 3.40 \\
     \hline
     \textbf{Rhythm+Pitch} & \textbf{Rhythm+Timbre} & \textbf{Pitch+Timbre}\\
     3.04 & 2.73 & 3.35 \\
     \hline
     \textbf{All Three} & \textbf{\textsc{AutoVC}} & \textbf{Source}  \\
     2.79 & 3.24 & 4.65 \\
     \hline\hline
\end{tabular}
\end{table}

Figure~\ref{fig:subject} shows the conversion rates of different types of conversions. As expected, the conversion rate is high when the corresponding aspect is converted, and low otherwise. For example, the pitch-only conversion has a high pitch conversion rate but low rhythm and timbre conversion rates; whereas the rhythm+timbre conversion has a high rhythm and timbre conversion rates but a low pitch conversion rate. It is worth noting that \textsc{AutoVC} has a high timbre conversion rate, but low in the other, indicating that it only converts timbre. In short, both the visualization results and our subjective evaluation verifies that each conversion can successfully convert the intended aspects, without altering the other aspects, whereas \textsc{AutoVC} only converts timbre.

We also evaluate the MOS (mean opinion score), ranging from one to five, on the quality of the conversion, as shown in table~\ref{mos}. There are a few interesting observations. First, the MOS of pitch conversion is higher than that of timbre and rhythm conversions, which implies that timbre and rhythm conversions are the more challenging tasks. Second, as the number of converted aspects increases, the MOS gets lower, because the conversion task gets more challenging.

\textbf{Objective Evaluation} ~Due to the lack of explicit labels of the speech components, it is difficult to fully evaluate the disentanglement results using objective metrics. However, we can still objectively evaluate the pitch-only conversion performance by comparing the pitch contour of the converted speech and the target pitch contour. Following \citet{valle2020mellotron}, we use three metrics for the comparison: Gross Pitch Error (GPE) \citep{nakatani2008method}, Voice Decision Error (VDE) \citep{nakatani2008method}, and F0 Frame Error (FFE) \citep{chu2009reducing}. \algname achieves a GPE of $1.04\%$, a VDE of $8.14\%$, and an FFE of $8.86\%$. As a reference, these results are comparable with the results reported in \citet{valle2020mellotron}, with a slightly higher GPE and lower VDE and FFE. Note that these two sets of results cannot be directly compared, because the datasets are different, but they show the effectiveness of the \algname in disentangling pitch.

\subsection{Mismatched Conversion Target}

Since utterances with mismatched contents have different numbers of syllables and lengths, we would like to find out how \algname converts rhythm when the source and target utterances read different content. Figure~\ref{fig:mismatched} shows the rhythm-only conversion between a long utterance, \textit{`And we will go meet her Wednesday'} (top left panel), and a short utterance, `\textit{Please call Stella'} (top right panel).

The short to long conversion is shown in the bottom left panel. It can be observed that the conversion tries to match the syllable structure of the long utterance by stretching its limited words. In particular, \textit{`please'} is stretched to cover \textit{`and we will'}, \textit{`call'} to cover \textit{`go meet'}, and \textit{`Stella'} to cover \textit{`her Wednesday'}. On the contrary, the long to short conversion, as shown in the bottom right panel, tries to squeeze everything to the limited syllable slots in the short utterance. Intriguingly still, the word mapping between the long utterance and the short utterance is \emph{exactly the same} as in the short to long conversion. In both cases, the word boundaries between the converted speech and the target speech are surprisingly aligned.

These observations suggest that \algname has an intricate \textit{`fill in the blank'} mechanism when combining the rhythm information with content and pitch. The rhythm code provides a number of blanks, and the decoder fills the blanks with the content information and pitch information provided by the respective encoders. Furthermore, there seems to be an anchoring mechanism that associates the content and pitch with the right blank, which functions stably even if the blanks and the content are mismatched.

\subsection{Removing Speech Components}
\label{subsec:zero}

\label{subsec:mismatched}

\begin{figure}[t]
\centering
\includegraphics[width=1\linewidth]{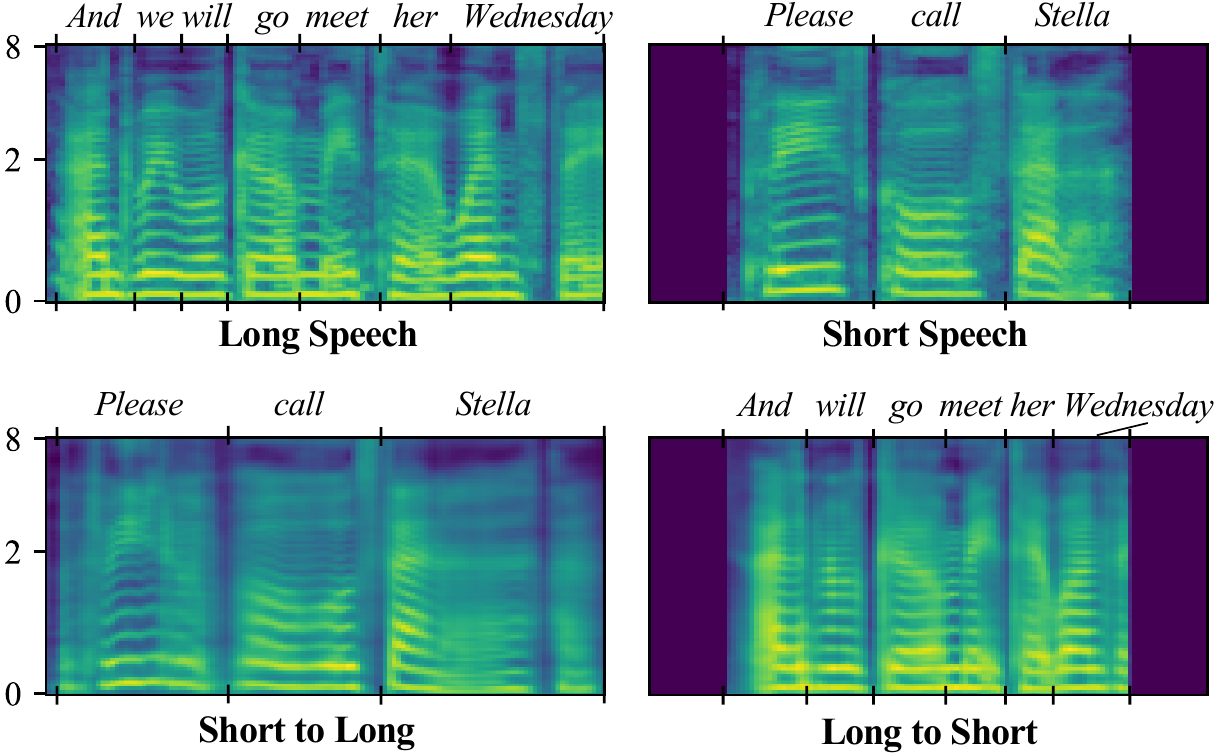}
\small
\vspace{-0.2in}
\caption{\small{Rhythm-only conversion when the source and target speech have mismatched content.}}
\label{fig:mismatched}
\end{figure}

\begin{figure}[t]
\centering
\includegraphics[width=1\linewidth]{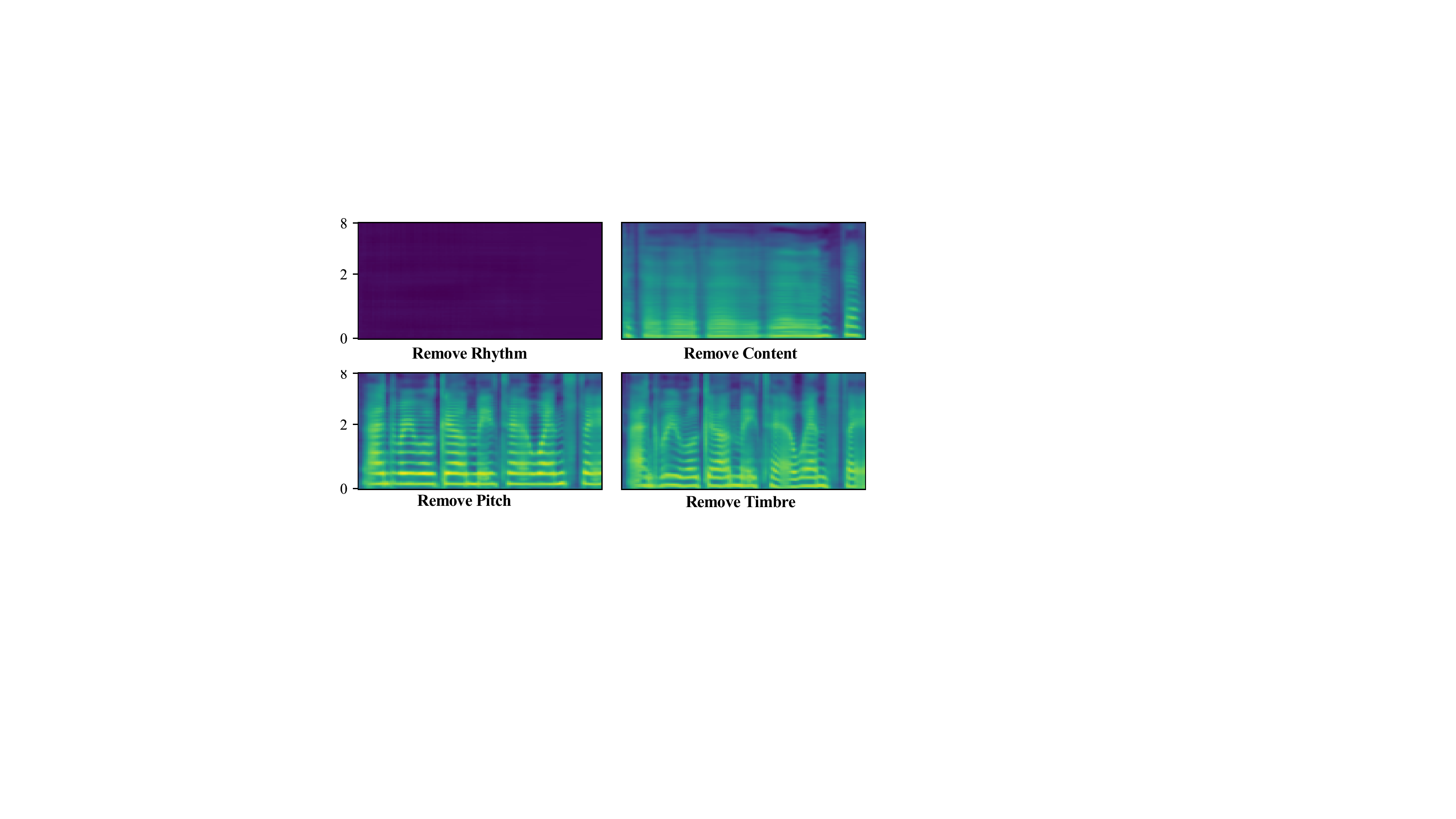}
\small
\vspace{-0.2in}
\caption{\small{Reconstructed speech when one speech component is removed. The ground truth speech is in figure~\ref{fig:mismatched} top left panel.}}
\label{fig:zero}
\end{figure}

To further understand the disentanglement mechanism of \algnamens, we generate spectrograms with one of the four components removed. To remove rhythm, content or pitch, we respectively set the input to the rhythm encoder, content encoder or pitch encoder to zero. To remove timbre, we set the speaker embedding to zero. Figure~\ref{fig:zero} shows the output spectrograms with one component removed. As can be observed, when the rhythm is removed, the output becomes zero, and when the content is removed, the output becomes a set of slots with no informative spectral shape. These findings are consistent with our `fill in the blank' hypothesis in section~\ref{subsec:mismatched}. When rhythm code is removed, there is no slot to fill, and hence the output spectrogram is blank. When content is removed, there is nothing to fill in the blanks, resulting in a spectrogram with uninformative blanks. When the pitch is removed, the pitch of the output becomes completely flat, as can be seen from the flat harmonics. Finally, when timbre is removed, the formant positions of the output spectrogram shift, which indicates that the timbre has changed, possibly to an average speaker. These results further verify that \algname can separately model different speech components.

\subsection{Varying the Information Bottleneck}
\label{subsec:vary_bottleneck}

In this section, we would like to verify our theoretical explanation in section~\ref{subsec:why} by varying the information bottleneck and see if \algname will still act as our theory predicts.

According to figure~\ref{fig:model}, if the physical information bottleneck of the rhythm encoder is too wide, then the rhythm encoder will pass all the information through, and the content encoder, pitch encoder and speaker identity will be useless. As a result, the rhythm-only conversion will convert all the aspects. On the other hand, the pitch-only and timbre-only conversions will alter nothing. Similarly, if the physical information bottleneck of the content encoder is too wide, but random sampling is still present, then the content encoder will pass almost all the information through, except for the rhythm information, because the random resampling operations still contaminate the rhythm information and \algname would still rely on the rhythm encoder to recover the rhythm information. As a result, the rhythm-only conversion would still convert rhythm, but the pitch-only and timbre-only conversions would barely alter anything.

Figure~\ref{fig:subject_wide} shows the subjective conversion rates of single-aspect conversions when the physical bottleneck of rhythm encoder or the content encoder is too wide. These results agree with our theoretical predictions. 
When the rhythm encoder physical bottleneck is too wide, the rhythm-only conversion converts all the aspects, while other conversions convert nothing. When the content encoder physical bottleneck is too wide, the rhythm-only conversion still converts rhythm. Notably, the timbre-only conversion still converts timbre to some degree, possibly due to the random resampling operation of the content encoder. These results verify our theoretical explanation of \algnamens.

%% file: sections/appendix.tex
\clearpage
\section{Proof to Theorem \ref{thm:main}}

The proof is divided into five parts.
\begin{lemma}
Under the assumptions in theorem~\ref{thm:main}, the global minimum of equation~\eqref{eq:loss} is \e{0}.
\end{lemma}
\begin{proof}
Construct the encoders that satisfy equation \eqref{eq:disentagle}, which is a feasible choice. Construct the decoder as follows:
\begin{equation}
    \small
    \begin{aligned}
    &D(\bm Z_c, \bm Z_r, \bm Z_f, \bm V) \\
    = &g_s(h_c^{-1}(\bm Z_c), h_r^{-1}(\bm Z_r), h_f^{-1}(\bm Z_f), g_u^{-1}(\bm V)) \\
    = &g_s(\bm C, \bm R, \bm F, U) \\
    = &\bm S,
    \end{aligned}
\end{equation}
which achieves \e{0} reconstruction loss in equation~\eqref{eq:loss}.
\end{proof}

\begin{lemma}
Equation~\eqref{eq:resample1} implies
\begin{equation}
\small
    I(\bm R; A(\bm S), f(\bm R)) < H(\bm R), \quad \forall f(\cdot) \mbox{ s.t. } H(\bm R|f(\bm R))>0.
    \label{eq:resample3}
\end{equation}
\begin{proof}
We will prove this by contradiction. If there exists an \e{f(\cdot)} s.t. \e{H(\bm R|f(\bm R))>0} but \e{I(\bm R; A(\bm S), f(\bm R)) = H(\bm R)}, then there exist \e{\bm r_1 \neq \bm r_2}, which \e{f(\cdot)} cannot distinguish but \e{A(\bm S)} can, \emph{i.e.}
\begin{equation}
\small
    A(g_s(\bm C_1, \bm r_1, \bm F_1, \bm V_1)) \neq A(g_s(\bm C_2, \bm r_2, \bm F_2, \bm V_2)), \quad \mbox{w.p. } 1.
\end{equation}
which contradicts with \eqref{eq:resample1}.
\end{proof}
\end{lemma}

\begin{lemma}
Under the assumptions in theorem~\ref{thm:main}, in order to achieve the global minimum of equation~\eqref{eq:loss}, \e{\bm Z_r} must satisfy equation~\eqref{eq:disentagle}.
\label{lem:r}
\end{lemma}
\begin{proof}
We will prove this by contradiction. If
\begin{equation}
    \small
    H(\bm R | \bm Z_r) > 0,
\end{equation}
then we have
\begin{equation}
    \small
    \begin{aligned}
    H(\bm R | \hat{\bm S}) &\geq H(\bm R | \bm Z_r, \bm Z_c, \bm Z_f, \bm V ) \\
    &\geq H(\bm R | \bm Z_r, \bm Z_c, \bm Z_f) \\
    &\geq H(\bm R | \bm Z_r, \bm A(\bm S), \bm A(\bm P)) \\
    & > 0,
    \end{aligned}
    \label{eq:ambiguous_r}
\end{equation}
where the first and third lines are due to the data processing inequality; the second line is given by equation~\eqref{eq:speech_gen} and the independence assumption among the aspects; the last line is given by equation~\eqref{eq:resample3}. Equation~\eqref{eq:ambiguous_r} essentially means \e{\hat{\bm S}} cannot reconstruct \e{\bm R}, and thereby cannot reconstruct \e{\bm S}, which contradicts with the optimal loss being \e{0}.

Moreover, if
\begin{equation}
    \small
    H(\bm Z_r | \bm R) > 0,
\end{equation}
then
\begin{equation}
    \small 
    H(\bm R | \bm Z_r) = H(\bm Z_r | \bm R) + H(\bm R) - H(\bm Z_r) \geq H(\bm Z_r | \bm R) > 0,
\end{equation}
where the '\e{\geq}' inequality is from equation \eqref{eq:info_assump}. This will again lead to a contradiction.
\end{proof}

\begin{lemma}
Under the assumptions in theorem~\ref{thm:main}, and assuming \e{\bm Z_r} satisfies equation~\eqref{eq:disentagle}, in order to achieve the global minimum of equation~\eqref{eq:loss}, \e{\bm Z_c} must satisfy equation~\eqref{eq:disentagle}.
\label{lem:c}
\end{lemma}
\begin{proof}
We will prove this by contradiction. If
\begin{equation}
    \small
    H(\bm C | \bm Z_c) > 0,
\end{equation}
then we have
\begin{equation}
\small
\begin{aligned}
    H(\bm C | \hat{\bm S}) &\geq H(\bm C | \bm Z_r, \bm Z_c, \bm Z_f) \\
    &= H(\bm C | f_r(\bm R), \bm Z_c, \bm Z_f) \\
    &= H(\bm C | \bm Z_c, \bm Z_f) \\
    &\geq H(\bm C | \bm Z_c, \bm F) \\
    &= H(\bm C | \bm Z_c, g_p(\bm F, \bm R)) \\
    &= H(\bm C | \bm Z_c) > 0,
\end{aligned}
\label{eq:ambiguous_c}
\end{equation}
where the first line is similar to equation~\eqref{eq:ambiguous_r}; the second line is given by \e{\bm R} satisfying equation~\eqref{eq:disentagle}; the third and last lines are due to the independence assumption among the aspects; the fourth line is given by the data processing inequality; the fifth line is given by equation~\eqref{eq:pitch_contour}. Equation~\eqref{eq:ambiguous_c} essentially means \e{\hat{\bm S}} cannot reconstruct \e{\bm C}, and thereby cannot reconstruct \e{\bm S}, which contradicts with the optimal loss being \e{0}.

Moreover, if
\begin{equation}
    \small
    H(\bm Z_c | \bm C) > 0,
\end{equation}
then
\begin{equation}
    \small 
    H(\bm C | \bm Z_c) = H(\bm Z_c | \bm C) + H(\bm C) - H(\bm Z_c) \geq H(\bm Z_c | \bm C) > 0,
\end{equation}
where the '\e{\geq}' inequality is from equation \eqref{eq:info_assump}. This will again lead to a contradiction.
\end{proof}

\begin{lemma}
Under the assumptions in theorem~\ref{thm:main}, and assuming \e{\bm Z_r} and \e{\bm Z_c} satisfy equation~\eqref{eq:disentagle}, in order to achieve the global minimum of equation~\eqref{eq:loss}, \e{\bm Z_f} must satisfy equation~\eqref{eq:disentagle}.
\label{lem:f}
\end{lemma}
\begin{proof}
We will prove this by contradiction. If
\begin{equation}
    \small
    H(\bm F | \bm Z_f) > 0,
\end{equation}
then we have
\begin{equation}
\small
\begin{aligned}
    H(\bm F | \hat{\bm S}) &\geq H(\bm F | \bm Z_r, \bm Z_c, \bm Z_f) \\
    &= H(\bm C | f_r(\bm R), f_c(\bm C), \bm Z_f) \\
    &= H(\bm C | \bm Z_f) > 0,
\end{aligned}
\label{eq:ambiguous_f}
\end{equation}
where the first line is similar to equation~\eqref{eq:ambiguous_r}; the second line is given by \e{\bm R} and \e{\bm C} satisfying equation~\eqref{eq:disentagle}; the third is due to the independence assumption among the aspects. Equation~\eqref{eq:ambiguous_f} essentially means \e{\hat{\bm S}} cannot reconstruct \e{\bm F}, and thereby cannot reconstruct \e{\bm S}, which contradicts with the optimal loss being \e{0}.

Moreover, if
\begin{equation}
    \small
    H(\bm Z_f | \bm F) > 0,
\end{equation}
then
\begin{equation}
    \small 
    H(\bm F | \bm Z_f) = H(\bm Z_f | \bm F) + H(\bm F) - H(\bm Z_f) \geq H(\bm Z_f | \bm F) > 0.
\end{equation}
where the '\e{\geq}' inequality is from equation \eqref{eq:info_assump}. This will again lead to a contradiction.
\end{proof}
Theorem~\ref{thm:main} can be implied by combining lemmas~\ref{lem:r}, \ref{lem:c} and \ref{lem:f}.

\section{Additional Implementation Details}
\label{sec:implementation_details}

\subsection{Input Features}

The input and output spectrograms are 80-dimensional mel-spectrograms computed using 64 ms frame length and 16 ms frame hop. For each speaker, the input pitch contour is first extracted using a pitch tracker \cite{yamamoto2019r9y9}, and then normalized by its mean and four times its standard deviation. This operation roughly limits the pitch contour to be within the range of 0-1. After that, we quantize the range 0-1 into 256 bins and turn it into one-hot representations. Finally, we add another bin to represent unvoiced frames producing 257 one-hot encoded feature \e{\bm P}. 

\subsection{Information Bottleneck Implementations}
\label{subsec:info_bottleneck}
As discussed, \algname adopts two methods to restrict the information flow. The first is random resampling, and the second is the constraints on the physical dimensions, which include the downsampling operations in frequency and time dimensions.

The random resampling is implemented as follows. First, the input signal is divided into segments, whose length is randomly uniformly drawn from 19 frames to 32 frames \cite{polyak2019attention}. Each segment is resampled using linear interpolation with a resampling factor randomly drawn from 0.5 (compression by half) to 1.5 (stretch). 
For each input utterance, the random sampling operations at the input layers of the content encoder and pitch encoder share the same random sampling factors. We find that by having the same random sampling factors, we can reduce the remaining rhythm information after the random sampling, and thus achieving better disentanglement.

We follow the downsampling implementation in \textsc{AutoVC}. Suppose the downsampling factor is \e{k} and we use zero-based indexing of the frames. For the forward direction output of the bidirectional-LSTM, \e{t=kn+k-1}, \e{n \in \{0, 1, 2 \cdots\}} are sampled; for the backward direction, \e{t= kn} are sampled. In this way, we can ensure the frames at both ends are covered by at least one forward code and one backward code.

\subsection{Converting Pitch}

During training, the input speech to the content encoder and the input pitch contour to the pitch encoder are always aligned (due to the shared random sampling factors), \emph{i.e.} they share the same (contaminated) rhythm information \e{A(\bm R)}. During pitch conversion, however, such alignment no longer exists, because the pitch contour is replaced with that of another utterance. To restore the temporal alignment, before we perform the pitch conversion, we first perform a rhythm-only conversion to the new pitch contour, where the conversion target is the input speech to the content encoder.

The rhythm-only conversion on pitch contour is essentially the same as the rhythm-only conversion on speech, except that we need to use a mini \algname variant that operates on pitch contour, not speech. Specifically, there are two major differences between the variant and the original \algnamens. First, the variant comes with only two encoders, the rhythm encoder and the pitch encoder, whose inputs are spectrograms and the corresponding pitch contours. The content encoder is removed because there is no content information in pitch contour. Second, rather than reconstructing speech, the decoder reconstructs pitch contour from the outputs of the encoders. The output dimension of the decoder at each time is the one-hot encoding dimension of the pitch contour (257), and the cross-entropy loss is applied. The hyperparameter settings are the same as in the original \algnamens. Following the same argument as in \algnamens, it can be shown that this variant can disentangle the pitch and rhythm information of pitch contour, and thus can perform the rhythm-only conversion.

\subsection{General Guide on Tuning the Bottlenecks}
\label{app:tuning}

Although tuning the information bottleneck dimensions is the most difficult part of getting \algname to work properly, there are some straightforward guidelines on how to choose the correct physical dimension of each code. The basic idea is that removing one of the four codes should reproduce the results in figure~\ref{fig:zero}.

Specifically, when the rhythm code is set to zero, the output should be almost blank, as shown in the top-left spectrogram in figure~\ref{fig:zero}. If the output still preserves significant speech energy, it means that the rhythm code dimension is too small. Consider increasing the dimension. Throughout this section, by increasing the dimension, we refer to two operations. The first operation is to increase the channel dimension of the encoder output. The second operation is to increase the sampling rate of the down-sampled code. Accordingly, by decreasing the dimension, we refer to the two opposite operations.

When the content code is set to zero, the output should become a set of slots with uninformative spectral shapes, as shown in the top-right spectrogram in figure~\ref{fig:zero}. If the output preserves \emph{the same} formant structure as the input speech, it means that the content code dimension is too small, and that the rhythm code dimension is too large (the case where the rhythm bottleneck is too wide has been discussed in section~\ref{subsec:vary_bottleneck}). In this case, consider increasing the content code dimension, and decreasing the rhythm code dimension. Please note that the key is to compare the formant structure with that of the input speech. In some cases, removing the content would produce a speech-like spectrogram with clear harmonic and formant structures, instead of the aforementioned empty slots of uninformative spectral shapes. However, as long as the formant structure is drastically different from the input speech, the bottleneck setting is appropriate.

When the pitch code is set to zero, the output should become either a voiced spectrogram with a constant pitch and harmonic structure, as shown is the bottom-left spectrogram in figure~\ref{fig:zero}, or an unvoiced spectrogram with no harmonic structure at all. If the output does not fall in either case, \emph{i.e.} the output preserves the same pitch or voiced/unvoiced states as the input speech, it means that either the rhythm code or the content code is too wide (both cases have been discussed in section~\ref{subsec:vary_bottleneck}). If this happens, determine which case it is by setting the content code to zero, and then make adjustments accordingly. It is also worth mentioning that it does not harm to set a relatively large bottleneck dimension for pitch, because the information conveyed by the pitch contour is already very limited.

Finally, if the speaker identity is changed to another speaker, the output should sound like the target speaker. If it does not sound like the target speaker, it means either the rhythm code or the content code is too wide. Follow the aforementioned procedures to identify the problem. However, if there are no anomalies in the aforementioned diagnosis, they may be both too wide. Try decreasing them simultaneously. Conversely, if the converted speech is of very poor quality, it implies that both the rhythm code and the content code are too narrow. Try increasing them simultaneously.

\subsection{Test Token Selection}
\label{app:test_select}

The samples for subjective evaluations on pitch are hand-picked by authors. Authors listen to all the pairs in the test set and find the parallel pairs that are perceptibly different in pitch, e.g. rise vs fall tones, or high vs low tones, in at least one word. We then sub-select 20 pairs that maintain speaker diversity. For rhythm, we identify top 40 parallel pairs with greatest differences in time length, and then sub-select 20 pairs that maintain speaker diversity. For timbre, we simply randomly pick pairs from different speakers. Please note that the selection is based on the original speech only. They are not based on any conversion results.

\section{Additional Experiment Results}
\label{sec:exper_add}

\begin{figure}[t]
\centering
\includegraphics[width=1\linewidth]{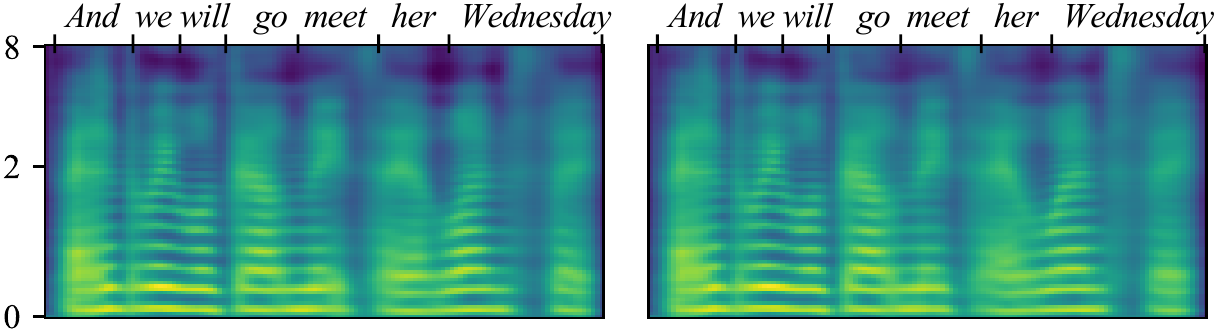}
\small
\vspace{-0.2in}
\caption{\small{Reconstructed speech produced by \textsc{AutoVC} with a random resampling module. The ground truth speech for the left column is in figure~\ref{fig:mismatched}.} The word boundaries and labels are copied from that of the ground truth.}
\label{fig:rr_recon}
\end{figure}

\subsection{Does Random Resampling Remove All Rhythm?}

In figure~\ref{fig:model} and section~\ref{subsec:why}, we assume that the random resampling only contaminates rhythm information, but does not completely remove it. To verify this assumption, we train a single autoencoder for speech, where the encoder and decoder are the \algname content encoder and decoder respectively. If random resampling only removes a portion of the rhythm information, the output reconstruction can still roughly temporally aligned with the ground truth speech. Otherwise, the reconstruction will be completely misaligned.

Figure~\ref{fig:rr_recon} shows two reconstruction results with different randomly drawn resampling factors, whose ground truth utterances are both the top-left panel of figure~\ref{fig:mismatched}. To assist our judgment of the alignment, we directly copy the word boundaries and labels from the ground truth. As can be observed, the two reconstructions are very alike, even though their random resampling factors are different. Furthermore, both reconstructions can recover the ground truth speech decently, only with some minor blurring, which verifies that random resampling performs an \emph{incomplete disentanglement} of rhythm. In other words, \algname shows that we can build a complete disentanglement mechanism even if we only have a partial disentanglement technique.

\subsection{Do Rhythm Labels Exist?}

\begin{figure}[t]
\centering
\includegraphics[width=1\linewidth]{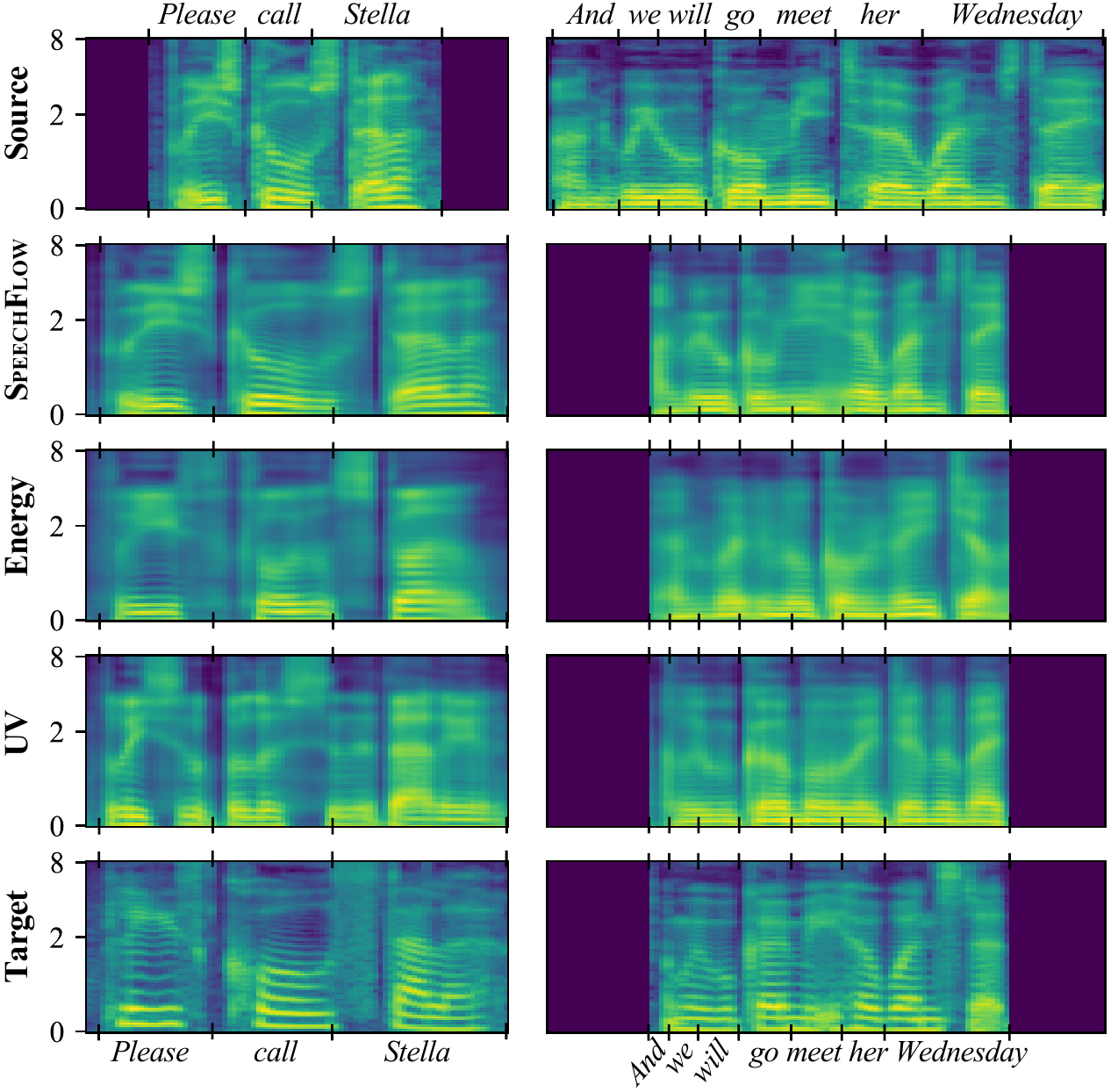}
\small
\caption{\small{Rhythm-only conversion using the rhythm feature in \algname (second row) compared with that using candidate rhythm features, including short-time energy (third row) and UV label (fourth row).}}
\label{fig:rhythm_feature}
\end{figure}

In section~\ref{sec:intro}, we have discussed that one motivation for designing \algname is that rhythm labels are not directly available. If they were, the rhythm aspect could be disentangled in much simpler ways. In this section, we would like to explore if there exist any rhythm labels.

We have identified two promising candidate rhythm labels, short-time energy and unvoiced-voiced (UV) label. The short-time energy is computed by taking the moving average of the squared waveform. The UV labels are derived from pitch contour, which equals one if the corresponding frame is voiced, and zero otherwise. Both candidates are informative of the syllable boundaries, and neither contains other information such as content and pitch. To test if these candidates are equally effective as the \algname rhythm encoder, we train two variants of \algnamens, one replacing the rhythm code with the short-time energy, and the other with the UV label. We then perform the rhythm-only conversion using \algname and the two variants, by replacing the rhythm code/label with that of the target speech. If the candidates are effective, the corresponding rhythm-only conversions should be successful.

Figure~\ref{fig:rhythm_feature} shows the rhythm-only conversion results on two utterances, \textit{`Please call Stella'} and \textit{`And we will go meet her Wednesday'}, produced by these three algorithms. At first glance, all the conversion results are temporally aligned with the target speech, which seems to suggest that the rhythm aspect has been successfully converted. However, a close inspection into the formant structure of the candidate conversion results reveals that the content within each syllable is completely incorrect.

With the `fill in the blank' perspective discussed in section~\ref{subsec:mismatched}, we can better understand why the candidate rhythm labels fail. Both candidates can accurately provide the temporal information of the syllable boundaries, and thus the blanks are correctly located in time. However, the candidates fail to provide the anchor information of what to fill in each blank, and that is why the conversion algorithms put the wrong content in the blanks. In summary, obtaining a rhythm label is a nontrivial task, because the rhythm label should contain some anchor information to associate each syllable with the correct content, while excluding excessive content to ensure content disentanglement. \algnamens, with a triple information bottleneck design, manages to obtain such an effective rhythm code, which contributes to a successful rhythm conversion.

\subsection{Additional Conversion Spectrograms}

\begin{figure*}[t]
\centering
\includegraphics[width=0.6\linewidth]{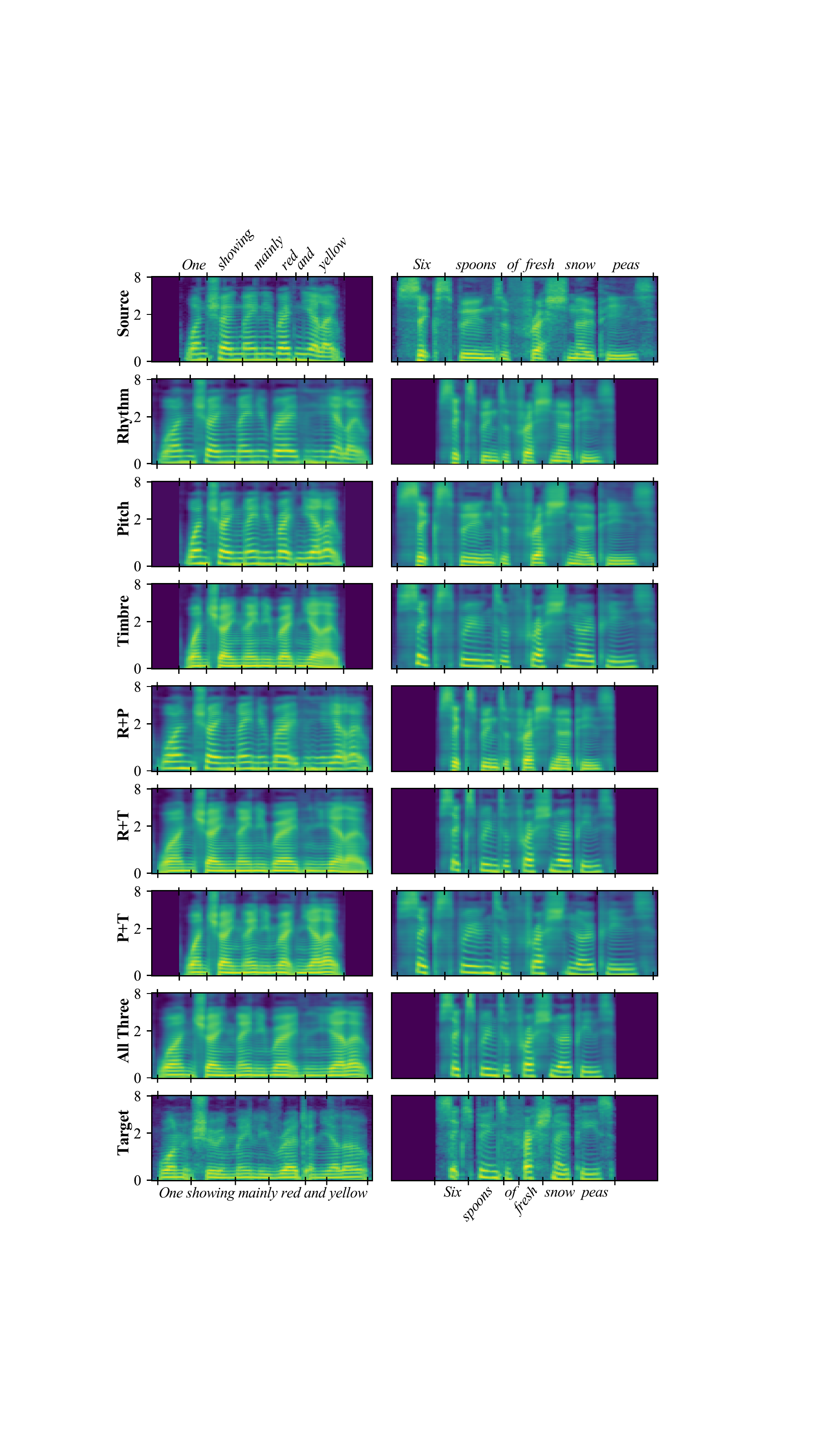}
\small
\caption{\small{Spectorgrams of aspect-specific conversion results on two utterances, \textit{`One showing mainly red and yellow'} (left) and \textit{`Six spoons of fresh snow peas'} (right). R+P denotes rhythm+pitch conversion; R+T denotes rhythm+timbre conversion; P+T denotes pitch+timbre conversion.}}
\label{fig:convert_plots}
\end{figure*}

In figure~\ref{fig:convert_plots}, we augment the spectrogram visualization results in section~\ref{subsec:main_conversion} (figure~\ref{fig:main_spect}) with two additional utterances, \textit{`One showing mainly red and yellow'} and \textit{`Six spoons of fresh snow peas'}, and with all the conversion types (not just the single-aspect conversions) displayed. Consistent with the results shown in section~\ref{subsec:main_conversion}, these additional results show that \algname can successfully convert the intended aspects to match those of the target speech, while keeping the remaining aspects matching the source speech. Remarkably, when all three aspects are converted, the converted speech becomes very similar to the target speech.